\documentclass[journal,draftcls,onecolumn,12pt,twoside]{IEEEtranTCOM}

\normalsize

\usepackage{amssymb, amsmath, color}
\usepackage{amsthm,graphicx,enumerate,verbatim,xcolor}
\usepackage[small]{caption}
\usepackage{subfigure}
\usepackage{tikz}

\newtheorem{thm}{Theorem}
\newtheorem{cor}{Corollary}
\newtheorem{lem}{Lemma}
\newtheorem{defn}{Definition}

\def\squarebox#1{\hbox to #1{\hfill\vbox to #1{\vfill}}}
\newcommand{\inout}         {{- \hspace{-0.3mm} \vbox{\hrule\hbox{\vrule\squarebox{.667em}\vrule}\hrule}\hspace{-0.3mm} -}}

\renewcommand{\baselinestretch}{1.76}

\begin{document}
\title{Rate-Distortion-Based Physical Layer Secrecy with Applications to Multimode Fiber}

\author{ Eva~C.~Song, ~\IEEEmembership{Student~Member,~IEEE,}
	Emina~Soljanin,~\IEEEmembership{Fellow,~IEEE,}
        Paul~Cuff, ~\IEEEmembership{Member,~IEEE,}
        H.~Vincent~Poor,~\IEEEmembership{Fellow,~IEEE,}
        and~ Kyle~Guan, ~\IEEEmembership{Member,~IEEE}
\thanks{E. C. Song, P. Cuff and H. V. Poor are with the Department
of Electrical Engineering, Princeton University, Princeton,
NJ 08544, USA e-mail: \{csong,cuff, poor\}@princeton.edu.}
\thanks{E. Soljanin and K. Guan are with Bell Labs, Alcatel-Lucent, Murray Hill,
NJ 07974, USA e-mail: \{emina,kyle.guan\}@alcatel-lucent.com}
}

\markboth{IEEE Transactions on Communications}%
{Submitted paper}

\maketitle

\begin{abstract}
Optical networks are vulnerable to physical layer attacks; wiretappers can improperly receive messages intended for legitimate recipients. Multimode fiber (MMF) transmission can be modeled via a broadcast channel in which both the legitimate receiver's and wiretapper's channels are multiple-input-multiple-output complex Gaussian channels. Our work considers the theoretical aspect of this security problem in the domain of a broadcast channel. Source-channel coding analyses based on the use of distortion as the metric for secrecy are developed. Alice has a source sequence to be encoded and transmitted over this broadcast channel so that the legitimate user Bob can reliably decode while forcing the distortion of the wiretapper, or eavesdropper, Eve's estimate as high as possible. Tradeoffs between transmission rate and distortion under two extreme scenarios are examined: the best case where Eve has only her channel output and the worst case where she also knows the past realization of the source. It is shown that under the best case, an operationally separate source-channel coding scheme guarantees maximum distortion at the same rate as needed for reliable transmission. Theoretical bounds are given, and particularized for MMF. Numerical results showing the rate distortion tradeoff are presented and compared with corresponding results for the perfect secrecy case.
\end{abstract}


\IEEEpeerreviewmaketitle

\section{Introduction}
Single mode fiber systems are believed to have reached their capacity limits. In particular, techniques such as wavelength-division multiplexing (WDM) and polarization-division multiplexing (PDM) have been heavily exploited in the past few years, leaving little room for further improvement in capacity \cite{winzer-express}. Space-division multiplexing (SDM) is a promising solution for meeting the growing capacity demands of optical communication networks. One way of realizing SDM is via the use of multimode fiber (MMF). While multimode transmission provides greater capacity, the security of such systems can be an issue because a wiretapper can eavesdrop upon MMF communication by simply bending the fiber \cite{ECOC} . 
MMF is a multiple-input-multiple-output (MIMO) system \cite{winzer-express} that captures the charateristics of crosstalk among different modes. The secrecy capacity of a Gaussian MIMO broadcast channel was studied in \cite{mimo-gaussian}, but the result cannot be applied directly to MMF because the channel is not the same. The secrecy capacity of this channel was studied in \cite{ECOC} where it is shown that the channel conditions required for perfect secrecy are quite demanding.  


The concepts of ``perfect secrecy", ``partial secrecy", ``strong secrecy", ``weak secrecy", ``equivocation" and ``distortion" will be applied repeatedly in this paper. We shall now briefly summarize the relationships among those terms. Please note that, even though ``perfect secrecy" is a non-asymptotic concept, here for convenience, we refer to ``perfect secrecy" in the asymptotic sense, i.e. the information leakage is arbitrarily small as the blocklength goes to infinity. Information theoretic secrecy typically considers one of the two regimes, perfect secrecy or partial secrecy. Perfect secrecy essentially requires no information leakage to the eavesdropper. In the regime of partial secrecy, one must quantify the degree of secrecy obtained. Equivocation rate is a metric found in the literature which measures how much of the signal is leaked to the eavesdropper regardless of whether the eavesdropper can use the leaked information in a constructive way. Another approach, taken in this work, is to use distortion to measure the difference between the original content and eavesdropper's estimate of it. 

The notions of strong and weak secrecy are both referring to the perfect secrecy regime, and equivocation is usually involved in the analysis. Under either strong or weak secrecy, distortion is at the maximum, because negligible information is leaked to the eavesdropper that would allow her to make a better estimate. However, implication in the other direction does not hold in general. In order to keep the distortion at a maximum, neither strong or weak secrecy is necessarily required. 

This work focuses on the partial secrecy regime. It should be noted that equivocation and distortion are not two independent measures. As pointed out in recent work \cite{schieler-itt}, equivocation becomes a special case of distortion when causal information is revealed to the eavesdropper and the distortion is measured by log loss. It can be seen from our results herein that high distortion can be achieved even if all the past information of the source is given to the eavesdropper. Furthermore, partial secrecy comes at a much lower cost than perfect secrecy. 

Distortion was also used in \cite{yamamoto} and \cite{cuff-globecom} as a metric for secrecy in the context of a noiseless network with secret key sharing. In this work, we are concerned with physical layer secrecy in MMF systems. This prompts us to formulate the problem as a source-channel coding problem along the lines studied in a general setting in \cite{allerton}, some results of which can be directly applied to MMF systems.

The rest of the paper is structured as follows. In Section \ref{system-model}, we introduce the system model in two ways: the general source-channel coding model for theoretical derivation; and the particular MMF channel model for application. In Section \ref{bounds}, we provide theoretical bounds with source-channel coding for general broadcast channels. This is our main theoretical contribution of the paper. In Section \ref{main-results}, we apply the general results from Section \ref{bounds} to the MMF model when the channel is time-invariant and discuss the secrecy outage in the case of a random channel in which the channel state information (CSI) is not available to the transmitter. In Section \ref{numerical}, we provide numerical evaluation to the MMF source-channel model under Hamming distortion. Finally, in Section \ref{conclusion}, we conclude the paper and discuss open problems from this work.

\section{System Model} \label{system-model}
We first introduce some notation that will be used throughout this paper. A sequence $X_1,..., X_n$ is denoted by $X^n$. Limit taken with respect to ``$n\rightarrow \infty$" is abbreviated as ``$\rightarrow_n$". In the case that $X$ is a random variable, $x$ is used to denote a realization and $\mathcal{X}$ is used to denote the support of that random variable. $\mathbb{R}$ and $\mathbb{C}$ are reserved to denote the real field and complex field, respectively. A complex Gaussian distribution is denoted by $\mathcal{CN}(\mu, \Sigma, C)$, where $\mu$ is the mean, $\Sigma$ is the covariance matrix, and $C$ is the relation matrix. A Markov relation is denoted by the symbol $\inout$. The total variation distance between two distributions $P$ and $Q$ are denoted by $||P-Q||_{TV}$. For a distortion measure $d: \mathcal{S} \times \mathcal{T}\mapsto \mathbb{R}^+$, the distortion between two sequences is defined to be the per-letter average distortion $d(s^k,t^k)=\frac1k\sum_{i=1}^k d(s_i,t_i)$. The maximum distortion $\Delta$, the average distortion achieved by guesses based only on the prior distribution of the source, is defined as
\begin{eqnarray} \label{max-distortion}
\Delta &\triangleq& \min_{t}\mathbb{E}[d(S,t)].
\end{eqnarray}

\subsection{Source-Channel Coding Model for General Broadcast Channel} \label{general-bcc}
A source node (Alice) has an independent and identically distributed (i.i.d.) sequence $S^k$ that she intends to transmit over a memoryless broadcast channel $P_{YZ|X}$ such that a legitimate user (Bob) can reliably decode the source sequence, while keeping the distortion between an eavesdropper (Eve) and Alice as high as possible. The source sequence $S^k$ is mapped to the channel input sequence $X^n$ through a source-channel encoder. Upon receiving $Y^n$, Bob makes an estimate $\hat S^k$ of the original source sequence $S^k$. Let $f_{k,n}: \mathcal{S}^k\mapsto \mathcal{X}^n$ be a source-channel encoder and $g_{k,n}: \mathcal{Y}^n\mapsto \mathcal{S}^k$ be the corresponding decoder. 
For almost lossless reconstruction, we require the probability that Bob's reconstruction differs from the original goes to zero asymptotically with the source blocklength. That is,  
$$\mathbb{P}\left[S^k\neq \hat S^k\right]\rightarrow_{k} 0.$$
Similarly, Eve also makes an estimate $T^k$ of $S^k$ upon receiving $Z^n$ and some other side information. We will examine two extreme cases based on the amount of side information Eve has.

\begin{figure}[htbp]
  \centering
  \includegraphics[width=9.6 cm]{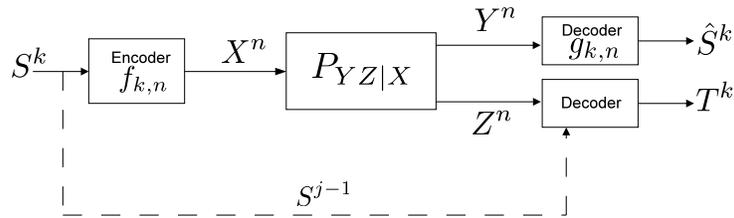}
\caption{Source-channel coding model with an i.i.d. source and broadcast channel. }
\label{sctotal}
\end{figure}

\subsubsection{No causal information available to Eve}
The case in which Eve has only her own channel output but no side information about the source corresponds to the best scenario for the legitimate users of the network, Alice and Bob. With the channel output alone, Eve has very limited resources in hand to make the estimate. 
Let $t^k$ be Eve's estimate of the original source sequence $s^k$. The system model is shown in Fig. \ref{sctotal}; however, the dashed line represents additional information that is not available to the eavesdropper in this first scenario. We use the lower case $t^k(z^n)$ to denote Eve's deterministic estimation functions of her observation $z^n$ and the capital letter $T^k=t^k(Z^n)$ to denote the function of  the random sequence $Z^n$. The following definitions in this section are for time-invariant channels. 

\begin{defn}
For a given distortion function $d(s,t)$, a rate distortion pair $(R,D)$ is achievable if there exists a sequence of encoder/decoder pairs $f_{k,n}$ and $g_{k,n}$ such that
$$\frac{k}{n}=R,$$
$$\lim_{n\rightarrow \infty} \mathbb{P}\left[S^k\neq\hat{S}^k\right]=0,$$
and
$$\liminf_{n\rightarrow \infty} \min_{t^k(z^n)}\mathbb{E}\left[d(S^k,t^k(Z^n))\right]\geq D.$$
\end{defn}
Note that the rate-distortion pair $(R,D)$ captures the tradeoff between Bob's rate for reliable transmission and Eve's distortion, which is different from rate-distortion theory in the traditional sense.

\subsubsection{With causal information available to Eve}
On the other hand, we are also interested in the case in which, at each time instance $j$, Eve gets to see the past realization of the source sequence ${S}^{j-1}$. This would be the worst scenario for the legitimate users. The definition for an achievable rate distortion pair $(R,D)$ is similar to Definition 1 given in the previous subsection except the last condition is replaced by
$$\liminf_{n\rightarrow \infty} \min_{\{t_j(z^n,s^{j-1})\}_{j=1}^k}\mathbb{E} \left[ \frac1k \sum_{j=1}^k d(S_j,t_j(Z^n,S^{j-1}) \right] \geq D.$$
The system model is shown in Fig. \ref{sctotal} with the dashed line representing the availability of the causal information. 

\subsection{MMF Channel Model} \label{mmf-model-section}
Now we particularize the general broadcast channel described above to an MMF broadcast channel as shown in Fig. \ref{mmf-channelmodel}.  An $M$-mode MMF is modeled as a memoryless MIMO channel with input $X$ an $M$-dimensional complex vector. Here $M$ is a positive integer.
\begin{figure}[htbp] 
  \centering
  \includegraphics[width=6 cm]{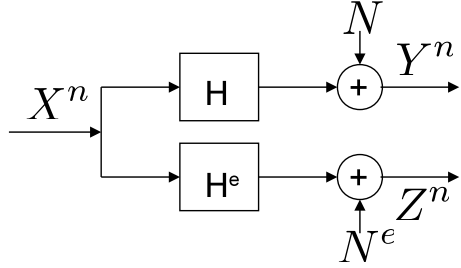}
\caption{MMF channel model }
\label{mmf-channelmodel}
\end{figure}

Unlike wireless MIMO which has a total power constraint, MMF channels have the following per mode power constraint averaged over $n$ uses of the channel:
\begin{equation}
\frac1n\sum_{i=1}^n\left|{ X_i^{(m)}}\right|^2\leq 1 \quad \text{for all modes} ~ m\in[1:M]. \label{power}
\end{equation}
More generally (as in \cite{mimo-gaussian}), we will consider a power constraint of the form
\begin{eqnarray}
\frac1n \sum_{i=1}^n{X}_i {X}_i^{\dagger}&\preceq& Q \label{powergen},
\end{eqnarray}
 where $Q\in \{A\in \mathcal{H}^{M\times M}: A\succeq 0, A_{ii}=1\}$ and $\mathcal{H}$ denotes the set of Hermitian matrices. One element in this set is the identity matrix $I$ (constraint (\ref{power})).  We will focus on the case that $Q=I$ for simplicity. A detailed discussion of the MMF channel model can be found in \cite{winzer-express}.

\subsubsection{The Legitimate User Communications Model}
The channel between Alice and Bob $P_{Y|X}$ is complex, Gaussian, MIMO, with input $X\in\mathbb{C}^M$ as described above, and output
$Y\in\mathbb{C}^M$ given by
\begin{eqnarray}
{Y}&=&H X+N, \label{bob}
\end{eqnarray}
where $ N\sim\mathcal{CN}(0,\sigma_N^2I,0)$ is $M$-dimensional, uncorrelated, zero-mean, complex, Gaussian noise and $H$ is an $M\times M$ complex matrix. Bob's channel matrix $H$ is of the form
\begin{eqnarray}
H&=&\sqrt{E_0L}\Psi, \label{h}
\end{eqnarray}
where $\Psi \in \mathbb{C}^{M\times M}$ is unitary and $E_0L$ is a constant scalar that measures the average power of the channel. We refer to
${E_0L}/{\sigma_N^2}$  as the SNR of the channel. Matrix $\Psi$, the unitary factor of the channel $H$, describes the modal crosstalk \cite{winzer-express}.


\subsubsection{The Eavesdropper Communications Model}
The channel between Alice and Eve $P_{Z|X}$ is also complex, Gaussian, MIMO, with input $X\in\mathbb{C}^M$ as described above, and output
$Z\in\mathbb{C}^M$ given by
\begin{eqnarray}
Z&=&H^e X+{N^e}, \label{eve}
\end{eqnarray}
where ${N^e}\sim\mathcal{CN}(0,\sigma_{N^e}^2 I,0)$ is $M$-dimensional uncorrelated, zero-mean, complex, Gaussian noise, and $H^e$ is an $M\times M$ complex matrix. Eve's channel matrix $H^e$ is of the form
\begin{eqnarray}
H^e&=&\sqrt{E_0L^e}\sqrt{\Phi}\Psi^e, \label{he}
\end{eqnarray}
where $\Psi^e\in \mathbb{C}^{M\times M}$ is unitary, $\Phi$ is diagonal with positive entries, and $E_0L^e$ is the average power of Eve's channel. Note that Eve has a different signal to noise ratio $\text{SNR}^e={E_0L^e}/{\sigma_{N^e}^2}$. The diagonal component $\Phi$ of the channel matrix $H^e$ corresponds to the mode-dependent loss (MDL) as introduced in \cite{winzer-express}. 


\section{Theoretical Bounds}\label{bounds}

In this section, we focus on the general broadcast channel introduced in Section \ref{general-bcc} only. We first make some general observations about the communication between Alice and Bob, as well as the communication between Alice and Eve. If Eve is not present, Alice and Bob can communicate losslessly at any rate lower than $R_0\triangleq\frac{\max_{X}I(X;Y)}{H(S)}$ because separate source-channel coding is optimal for point-to-point communication. Ideally, we want to force maximum distortion $\Delta$ upon Eve. But higher distortion to Eve may come at the price of a lower communication rate to Bob. The technical content of this section is organized as follows: the rate-distortion region for the ``no causal information" case is first given in Theorem \ref{nocausal}; to prepare for the achievability proof of Theorem \ref{nocausal}, an operational separation scheme is discussed; also, an achievable rate-distortion region is given in Theorem \ref{rdcausal} for the causal case under Hamming distortion; and finally, an example with a binary symmetric channel and Hamming distortion is provided for illustration.  

Before starting the new results, we shall provide a recap of what have been done in the literature regarding this problem and what our main advances are in this work. For noiseless channels, the source coding problems of both the no-causal-information and the causal-information cases were solved in \cite{schieler} and \cite{cuff-globecom}, respectively. There, secrecy was obtained by using a secret key shared between Alice and Bob. As for physical layer secrecy of a memoryless broadcast channel, the result for transmitting two messages, one confidential and one public, from Csisz\'{a}r and K\"{o}rner \cite{csiszar} have been known for many decades. In their work, weak secrecy were considered. This result was strengthened in \cite{raf} by considering strong secrecy. The same rate region was obtained in \cite{raf}, however the metric for secrecy is stronger. In our work, the source-channel coding schemes we propose operationally separate source and channel coding that require dividing the bit sequence produced by source coding into two messages which are then processed by the channel coding. The channel coding part functions in a way that is similar to \cite{csiszar} or \cite{raf}, except that the public message in their work is not required to be decoded in our case, and we refer to that message as the ``non-confidential" message. This type of channel coding setting was also used in \cite{allerton} for the causal-information case and it is shown that only weak secrecy is required to combine the source and channel coding. In this work, we will connect the source coding (from \cite{schieler}) and channel coding for the no-causal-information case. Unlike the causal-information case, strong secrecy from the channel is needed. This will require modifying some of the settings from \cite{raf}. We also extend the result for the causal-information case from \cite{allerton} to include the rate-distortion tradeoff.

We now state the rate-distortion result for general source-channel coding with an i.i.d. source sequence and a discrete memoryless broadcast channel $P_{YZ|X}$ when \textbf{no causal information} is available to Eve. In the following theorem, we will see that the source sequence can be delivered almost losslessly to Bob at a rate arbitarily close to $R_0$ while the distortion to Eve is kept at $\Delta$, as long as the secrecy capacity is positive. 
\begin{thm}\label{nocausal}
For an i.i.d. source sequence $S^k$ and memoryless broadcast channel $P_{YZ|X}$, if there exists $W\inout X\inout YZ$ such that $I(W;Y)-I(W;Z)>0$, then $(R,D)$ is achievable if and only if
\begin{align}
R&<\frac{\max_{X} I(X; Y)}{H(S)}, \label{ineq_r}\\
D&\leq \Delta, \label{ineq_d}
\end{align}
where $\Delta$ was defined in $(\ref{max-distortion})$.
\end{thm}
\textbf{Remark}: The requirement $I(W;Y)-I(W;Z)>0$ implies the existence of a secure channel with a positive rate, i.e. the eavesdropper's channel is not less noisy than the intended receiver's channel. So instead of demanding a high secure transmission rate with perfect secrecy to accommodate the description of the source, we need only to ensure the existence of a secure channel with positive rate. This will suffice to ensure the eavesdropper's distortion is maximal.

The converse is straightforward. Each of the inequalities $(\ref{ineq_r})$ and $(\ref{ineq_d})$ is true individually for any channel and source, $(\ref{ineq_r})$ by channel capacity coupled by optimality of source-channel separation, and $(\ref{ineq_d})$ by definition.
 
The idea for achievability is to operationally separate the source and channel coding (see Fig. \ref{separate}). The source encoder compresses the source and splits the resulting message into a confidential message and a non-confidential message. A channel encoder is concatenated digitally with the source encoder so that the channel delivers both the confidential and non-confidential messages reliably to Bob and keeps the confidential message secret from Eve, as in \cite{csiszar}. The overall source-channel coding rate will have the following form: $R=\frac{k}{n}=\frac{k}{\log|\mathcal{M}|}\cdot\frac{\log|\mathcal{M}|}{n}=\frac{R_{ch}}{R_{src}}$, where $|\mathcal{M}|$ is the total cardinality of the confidential and the non-confidential messages; $R_{ch}$ and $R_{src}$ are the channel coding and source coding rates, respectively. 

Let us look at two models in the following subsections that will help us establish the platform for showing the achievability of Theorem \ref{nocausal}.

\begin{figure}[htbp]
  \centering
  \includegraphics[width=12 cm]{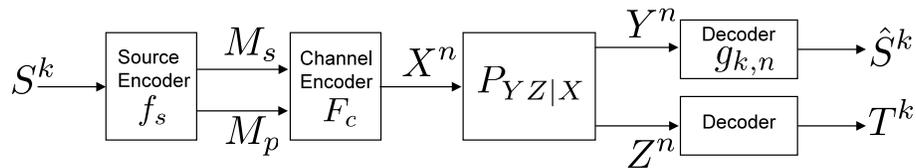}
\caption{Operational separate source-channel coding: the confidential and non-confidential messages satisfy $M_s\in[1:2^{kR_s'}=2^{nR_s}]$ and $M_p\in[1:2^{kR_p'}=2^{nR_p}]$}
\label{separate}
\end{figure}

\subsection{Channel Coding and Strong Secrecy} \label{channel}
Consider a memoryless broadcast channel $P_{YZ|X}$ and a communication system with a confidential message $M_s$ and a non-confidential message $M_p$ that must allow the intended receiver to decode both $M_s$ and $M_p$ while keeping the eavesdropper from learning anything about $M_s$. Problems like this were first studied by Csisz\'{a}r and K\"{o}rner \cite{csiszar} in 1978, as an extension of Wyner's work in \cite{wyner}. However, their model and our model differ in that the second receiver in their setting is required to decode the public message $M_p$. The mathematical formulation and result of our channel model is stated below.
We focus on the message pairs $(M_s, M_p)$ whose distribution satisfies the following:
\begin{eqnarray}
P_{M_s|M_p=m_p}(m_s)&=&2^{-nR_s} \label{uniform}
\end{eqnarray}
for all $(m_s,m_p)$. Later we will show a source encoder can always prepare the input messages to the channel of this form.

\begin{defn}
A $(R_s, R_p, n)$ channel code consists of a channel encoder $F_c$ (possibly stochastic) and a channel decoder $g_c$ such that
$$F_c: \mathcal{M}_s\times\mathcal{M}_p\mapsto \mathcal{X}^n$$
and
$$g_c: \mathcal{Y}^n\mapsto\mathcal{M}_s\times\mathcal{M}_p$$
where $|\mathcal{M}_s|=2^{nR_s}$ and $|\mathcal{M}_p|=2^{nR_p}$.
\end{defn}

\begin{defn} \label{def-weak}
The rate pair $(R_s, R_p)$ is achievable under weak secrecy if for all $(M_s, M_p)$ satisfying $(\ref{uniform})$, there exists a sequence of $(R_s, R_p, n)$ channel codes such that
$$\lim_{n\rightarrow \infty}\mathbb{P}\left[(M_s,M_p)\neq(\hat{M}_s,\hat{M}_p)\right]=0$$
and
$$\lim_{n\rightarrow \infty}\frac1n I(M_s;Z^n|M_p)=0.$$
\end{defn}
Note that because the eavesdropper may completely or partially decode $M_p$, the secrecy requirement is modified accordingly to consider $I(M_s;Z^n|M_p)$ instead of $I(M_s;Z^n)$. To guarantee true secrecy of $M_s$, we want to make sure that even if $M_p$ is given to the eavesdropper, there is no information leakage of $M_s$, because $I(M_s;Z^n|M_p)=I(M_s;Z^n,M_p)$ if $M_s$ and $M_p$ are independent.

\begin{thm}[Theorem 3 in \cite{allerton}] \label{region}
A rate pair $(R_s, R_p)$ is achievable under weak secrecy if
\begin{eqnarray}
R_s&\leq& I(W;Y|V)-I(W;Z|V),\\
R_p&\leq& I(V;Y)
\end{eqnarray}
for some $V\inout W \inout X\inout YZ$.
\end{thm}
The proof can be found in \cite{allerton}. Let us denote the above region as $\mathcal{R}$. We now strengthen the result by considering strong secrecy introduced in \cite{maurer}. Later we will use strong secrecy to connect the operationally separate source and channel encoders.

\begin{defn} \label{def-strong}
The rate pair $(R_s, R_p)$ is achievable under strong secrecy if for all $(M_s, M_p)$ satisfying $(\ref{uniform})$, there exists a sequence of $(R_s, R_p, n)$ channel codes such that
$$\lim_{n\rightarrow \infty}\mathbb{P}[(M_p,M_s)\neq(\hat{M}_s,\hat{M}_p)]=0$$
and
$$\lim_{n\rightarrow \infty} I(M_s;Z^n|M_p)=0.$$
\end{defn}
In general, weak secrecy does not necessarily imply that strong secrecy is also achievable; however, in this particular setting we have the following claim:
\begin{thm}\label{thmstrong}
A rate pair $(R_s, R_p)$ achievable under weak secrecy is also achievable under strong secrecy.
\end{thm}
The following two lemmas will assist the proof of Theorem \ref{thmstrong} by providing a sufficient condition for satisfying the secrecy constraint $\lim_{n\rightarrow\infty}I(M_s;Z^n|M_p)=0$.
\begin{lem}\label{lem2}
If $||P_{Z^n|M_p=m_p}P_{M_s|M_p=m_p}-P_{Z^nM_s|M_p=m_p}||_{TV}\leq\epsilon\leq \frac12$, then
$$I(M_s;Z^n|M_p=m_p)\leq -\epsilon\log\frac{\epsilon}{|\mathcal{M}_s|}.$$
\end{lem}
\begin{proof}
Let $\epsilon_{z^n}=\left|\left| P_{M_s|M_p=m_p}-P_{M_s|Z^n=z^n,M_p=m_p}\right|\right|_{TV}$. Therefore, 
$$\mathbb{E}_{P_{Z^n|M_p=m_p}}[\epsilon_{z^n}]=\left|\left|P_{Z^n|M_p=m_p}P_{M_s|M_p=m_p}-P_{Z^nM_s|M_p=m_p}\right|\right|_{TV}\leq \epsilon.$$
By Lemma 2.7 \cite{ckbook}, $\left|H(M_s|M_p=m_p)-H(M_s|Z^n=z^n,M_p=m_p)\right|\leq -\epsilon_{z^n}\log{\frac{\epsilon_{z^n}}{|\mathcal{M}_s|}}.$ Note that $f(x)\triangleq-x\log x$ is concave. And by applying Jensen's inequality twice, we have
\begin{eqnarray*}
&&I(M_s;Z^n|M_p=m_p)=\left|\mathbb{E}_{P_{Z^n|M_p=m_p}}\left[H(M_s|M_p=m_p)-H(M_s|Z^n=z^n,M_p=m_p)\right] \right|\\
&\leq&\mathbb{E}_{P_{Z^n|M_p=m_p}}\left[\left|H(M_s|M_p=m_p)-H(M_s|Z^n=z^n,M_p=m_p)\right|\right]\\
&\leq&\mathbb{E}_{P_{Z^n|M_p=m_p}}\left[-\epsilon_{z^n}\log{\frac{\epsilon_{z^n}}{|\mathcal{M}_s|}}\right]\\
&\leq& -\epsilon\log{\frac{\epsilon}{|\mathcal{M}_s|}}.
\end{eqnarray*}
\end{proof}
\begin{lem}\label{lem1}
If for every $(m_s,m_p)$, there exists a measure $\theta_{m_p}$ on $\mathcal{Z}^n$ such that
$$||P_{Z^n|M_p=m_p, M_s=m_s}-\theta_{m_p}||_{TV}\leq \epsilon_n$$
then
$$\lim_{n\rightarrow \infty}I(M_s;Z^n|M_p)=0$$
where $\epsilon_n=2^{-n\beta}$ for some $\beta>0$.
\end{lem}
A proof of Lemma \ref{lem1} is given in Appendix \ref{applem1}.

If there exist channel codes such that $\mathbb{P}\left[(M_s,M_p)\neq(\hat{M}_s,\hat{M}_p)\right]\rightarrow_n 0$ and measure $\theta_{m_p}$ for all $(m_s,m_p)$ such that $||P_{Z^n|M_p=m_p, M_s=m_s}-\theta_{m_p}||_{TV}\leq \epsilon_n$, then Theorem \ref{thmstrong} follows immediately. The existence of such a code and measure is assured by the same codebook construction and choice of measure as in \cite{raf}.

\subsection{Source Coding} \label{source}
Recall from our problem setup in Section \ref{system-model} that the sender Alice has an i.i.d. source sequence $S^k$. A source encoder is needed to prepare $S^k$ by encoding it into a pair of messages $(M_s, M_p)$ that satisfies $P_{M_s|M_p=m_p}(m_s)=2^{-kR_s'}=2^{-nR_s}$ so that it forms a legitmate input to the channel model in Section \ref{channel}.

\begin{defn}
An $(R_s', R_p', k)$ source code consists of an encoder $f_s$ and a decoder $g_s$ such that
$$f_s: \mathcal{S}^k \mapsto \mathcal{M}_s \times \mathcal{M}_p$$
$$g_s: \mathcal{M}_s\times\mathcal{M}_p \mapsto \mathcal{S}^k$$
where $|\mathcal{M}_s|=2^{kR_s'}$ and $|\mathcal{M}_p|=2^{kR_p'}$.
\end{defn}

\begin{defn}
A rate distortion triple $(R_s',R_p',D)$ is achievable under a given distortion measure $d(s,t)$ if there exists a sequence of $(R_s',R_p',k)$ source codes such that
$$\lim_{k\rightarrow \infty}\mathbb{P}\left[S^k\neq g_s(f_s(S^k))\right]=0$$
and the message pair generated by the source encoder satisfies $P_{M_s|M_p=m_p}(m_s)=2^{-kR_s'}$ and for all $P_{Z^n|M_sM_p}$ such that $I(M_s;Z^n|M_p)\rightarrow_n 0$
$$\liminf_{k\rightarrow\infty}\min_{t^k(z^n)}\mathbb{E}\left[d^k(S^k,t^k(Z^n))\right]\geq D.$$
\end{defn}

\begin{thm}\label{sourcecoding}
$(R_s',R_p',D)$ is achievable if 
$$R_s'>0,$$
$$R_s'+R_p'> H(S),$$
and
$$D\leq \Delta.$$
\end{thm}

The general idea for achievability is to consider the $\epsilon$-typical $S^k$ sequences and partition them into bins of equal size so that each bin contains sequences of the same type. The identity $M_p$ of the bin is revealed to all parties, but the identity $M_s$ of each sequence inside a bin is perfectly protected.  \footnote[1]{Strictly speaking, the source encoder may violate the condition $(\ref{uniform})$ on $(k+1)^{|\mathcal{S}|}$ number of bins, because $(k+1)^{|\mathcal{S}|}$ is an upper bound on the number of types of sequence with length $k$. However, this is just a very small (polynomial in $k$) number of bins compared with the total number (roughly $2^{kH(S)}$) of bins. Therefore, for this small portion of ``bad" bins that violates $(\ref{uniform})$, we can just let the source encoder declare an error on the confidential message $M_s$ and constructs a dummy $M_s$ uniformly given the bin index $m_p$. This will contribute only an $\epsilon$ factor to the error probability. } Each of such partitions is treated as a codebook. It was shown in \cite{schieler} that, for the noiseless case in which Eve is given $m_p$ instead of $z^n$, the distortion averaged over all such codebooks achieves the maximum distortion $\Delta$ as $k\rightarrow \infty$ and therefore there must exist one partition that achieves $\Delta$. In order to transition from the result in \cite{schieler} to our claim in Theorem \ref{sourcecoding}, we only need to show 
$$\min_{t^k(z^n)}\mathbb{E}\left[d^k(S^k,t^k(Z^n))\right]\geq\min_{t^k(m_p)}\mathbb{E}\left[d^k(S^k,t^k(M_p))\right].$$
\begin{IEEEproof}
First, observe that
\begin{eqnarray}
\min_{t^k(\cdot)} \mathbb{E}\left[d^k(S^k,t^k(\cdot))\right]&=&\frac1k \sum_{i=1}^k  \min_{t(i,\cdot)}\mathbb{E}\left[d(S_i, t(i,\cdot))\right] \label{equality}
\end{eqnarray}

Next, we claim the channel output sequence $z^n$ does not provide Eve anything more than $m_p$ and therefore
\begin{eqnarray}
\min_{t(i,z^n)}\mathbb{E}\left[\frac1k\sum_{i=1}^k d(S_i,t(i,Z^n))\right]&\geq&\min_{t(i,m_p)}\mathbb{E}\left[\frac1k\sum_{i=1}^k d(S_i,t(i,M_p))\right]-2\delta^\prime(\epsilon) \label{resultfin}
\end{eqnarray}
The analysis is similar to that in \cite{allerton}, but for the sake of clarity, we present the complete proof of $(\ref{resultfin})$ in Appendix \ref{detail}. Here strong secrecy comes into play. This is also pointed out within the proof in Appendix \ref{detail} that $I(M_s;Z^n|M_p)\rightarrow_n 0$ is needed.

Finally, combining $(\ref{resultfin})$ with $(\ref{equality})$ give us the desired result. 
\end{IEEEproof}

\subsection{Achievability of Theorem \ref{nocausal}}
With all the elements from Section \ref{channel} and \ref{source}, we are now ready to harvest the achievability proof of Theorem \ref{nocausal} using Theorems \ref{region} and \ref{sourcecoding} by concatenating the channel encoder with the source encoder.
\begin{IEEEproof}
Fix $\nu\geq\epsilon>0$. Fix $P_S$. Let ${R_s}^\prime=2\nu$, ${R_p}^\prime=H(S)-\nu$ and $R'={R_s}^\prime+{R_p}^\prime$. We apply the same codebook construction and encoding scheme as in Section \ref{source} by partioning the $\epsilon$-typical $S^k$ sequences into $2^{k{R_p}^\prime}$ bins and inside each bin we have $2^{k{R_s}^\prime}$ sequences so that $\mathbb{P}[S^k\neq g_s(f_s(S^k))]\leq \epsilon$. Recall that all the sequences inside one bin are of the same type, so it is guaranteed that
$$P_{M_s|M_p=m_p}(m_s)=\frac{1}{|\mathcal{M}_s|}=\frac{1}{2^{kR_s^\prime}}$$
for all $m_p$, $m_s$, which implies $I(M_s;M_p)=0$.

Let $R_s$ and $R_p$ be the channel rates. $R_p$ is seen as a function of $R_s$ on the boundary of the region given in Theorem \ref{region} and this is denoted by $R_p(R_s)$. Suppose $\max_{(R_s, R_p)\in \mathcal{R}}R_s>0$, i.e. there exists $W\inout X\inout YZ$ such that $I(W;Y)-I(W;Z)>0$ (justified in Appendix \ref{just}). $R_p(R_s)$ is continuous and non-increasing. Thus, $R_p$ achieves the maximum at $R_s=0$, which would be the channel capacity $\max_X I(X;Y)$ of $P_{Y|X}$ for reliable transmission. By the continuity of $R_p(R_s)$, $(R_s,R_p)=(2\nu\frac{k}{n},R_p(0)-\delta(\nu))$ is achievable under strong secrecy, i.e. $\mathbb{P}[(M_s, M_p)\neq(\hat{M}_s, \hat{M}_p)]\leq \epsilon$ and $I(M_s;Z^n|M_p)\leq \epsilon$, where $\delta(\nu)\rightarrow 0$ as $\nu\rightarrow 0$.

From the above good channel code under strong secrecy we have  $P_{Z^n|M_sM_p}$ such that $I(M_s;Z^n|M_p)\rightarrow_n 0$. Therefore, we can apply Theorem \ref{sourcecoding} to achieve 
$$\liminf_{k\rightarrow \infty}\min_{t^k(z^n)}\mathbb{E}\left[d^k(S^k,t^k(Z^n))\right]=D.$$

The error probability is bounded by the sum of the error probabilities from the source coding and channel coding parts i.e. $\mathbb{P}\left[S^k\neq \hat{S}^k\right]<2\epsilon$. Finally, we verify the total transmission rate to complete the proof:
\begin{eqnarray*}
R&=&\frac{k}{n}=\frac{R_s+R_p}{R_s^\prime+R_p^\prime}\\
&=&\frac{R_p(0)-\delta(\nu)+2R\nu}{H(S)+\nu}\\
&\geq&\frac{R_p(0)-\delta(\nu)}{H(S)+\nu}\\
&\stackrel{\nu\rightarrow 0}{\longrightarrow}&\frac{\max_X I(X;Y)}{H(S)}.
\end{eqnarray*}
\end{IEEEproof}

We next state the rate-distortion result for source-channel coding with an i.i.d. source sequence and discrete memoryless broadcast channel $P_{YZ|X}$ when \textbf{causal information} is available to Eve. The result comes from the rate matching of \cite{allerton}.
\begin{thm}\label{rdcausal}
For an i.i.d. source sequence $S^k$ and a memoryless broadcast channel $P_{YZ|X}$, a rate distortion pair $(R,D)$ is achievable if
\begin{eqnarray*}
R&\leq&\min\left(\frac{I(V;Y)}{I(S;U)}, \frac{I(W;Y|V)-I(W;Z|V)}{H(S|U)}\right),\\
D&\leq&\frac{\alpha}{R}\cdot \Delta+\left(1-\frac{\alpha}{R}\right)\cdot\min_{t(u)}\mathbb{E}\left[d(S,t(U))\right]
\end{eqnarray*}
for some distribution $P_SP_{U|S}P_VP_{W|V}P_{X|W}P_{YZ|X}$, where
$\alpha=\frac{[I(V;Y)-I(V;Z)]^+}{I(S;U)}$.
\end{thm}

\subsection*{Example: binary symmetric broadcast channel (BSBCC) and binary source with Hamming distortion}
To visualize Theorem \ref{nocausal} and Theorem \ref{rdcausal}, we will illustrate the results with a BSBCC and binary source under Hamming distortion, defined as
\begin{displaymath}
   d_H(s,t) = \left\{
     \begin{array}{lr}
       0  , \ s=t,\\
       1  , \text{ otherwise.}
     \end{array}
   \right.
\end{displaymath}
With the above setting, suppose $S_i \sim$ Bern$(p)$, and the broadcast channel is binary symmetric with crossover probabilities to the intended receiver and the eavesdropper $p_1$ and $p_2$, respectively. Assume $p\leq 0.5$ and $p_1< p_2<0.5$.  It is well known that this can be treated as a physically degraded channel in capacity calculation. Let us make the following definitions:
\begin{eqnarray}
&&f(x) \text{ is the linear interpolation of the points } \left(\log n,\frac{n-1}{n}\right), n=1,2,3,...\label{def_f}\\
&&d(x) \triangleq \min(f(x), 1-\max_s P_S(s))\label{def_d},\\
&&h(x)\triangleq x\log\frac{1}{x}+(1-x)\log\frac{1}{1-x} \text{ is the binary entropy function,} \label{def_binentropy}\\
&& x_1*x_2 \triangleq x_1*(1-x_2)+(1-x_1)*x_2 \text{ is the binary convolution,} \label{def_binconv}
\end{eqnarray}
where $P_S(\cdot)$ is the probability mass function of the random variable $S$.
The corresponding rate-distortion regions for the no-causal-information and causal-information cases are given in the following corollaries.
\begin{cor} \label{bsc-nocausal}
For an i.i.d. Bern$(p)$ source sequence $S^k$ and BSBCC with crossover probabilities $p_1$ and $p_2$, when \textbf{no causal information} is available, $(R,D)$ is achievable if and only if
\begin{eqnarray*}
R&<&\frac{1-h(p_1)}{h(p)},\\
D&\leq& p.
\end{eqnarray*}
\end{cor}

\begin{cor} \label{cor-causal}
For an i.i.d. Bern$(p)$ source sequence $S^k$ and BSBCC with crossover probabilities $p_1$ and $p_2$, when \textbf{causal information} is available, $(R,D)$ is achievable if
\begin{eqnarray*}
&&R\leq \frac{h(p_2)-h(p_1)}{h(p)},\\
&&D\leq p\\
\text{or}\\
&&\frac{h(p_2)-h(p_1)}{h(p)}< R \leq \frac{1-h(p_1)}{h(p)} ,\\
&&D\leq \alpha'p+(1-\alpha') d\left(\frac{h(\gamma*p_1)-h(\gamma*p_2)-h(p_1)+h(p_2)}{R}\right)
\end{eqnarray*}
where $\gamma\in[0,0.5]$ solves $h(\gamma*p_2)=1-h(p_1)+h(p_2)-Rh(p)$ and $\alpha'=\frac{h(\gamma*p_2)-h(\gamma*p_1)}{1-h(\gamma*p_1)}$.
\end{cor}
These corollaries result directly from applying Theorem \ref{nocausal} and Theorem \ref{rdcausal}, respectively. The region given in Corollary \ref{cor-causal} is calculated in a similar fashion as the region given by Theorem 7 of \cite{allerton}. An numerical example with $p=0.3$, $p_1=0.1$ and $p_2=0.2$ is plotted in Fig.\ref{bsc-ex}. Interpretation of the plot is deferred until the end of Section \ref{numerical}.

\begin{figure}[htbp]
  \centering
  \includegraphics[width=10 cm]{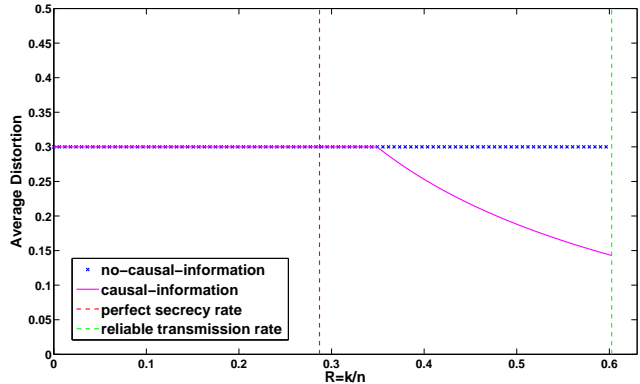}
\caption{Achievable distortion-rate curves. On the horizontal axis is the symbol/channel use source-channel coding rate and on the vertical axis is the average Hamming distortion.} 
\label{bsc-ex}
\end{figure}

\section{MMF Main Results} \label{main-results}
\subsection{Fixed MMF Channel}
We now apply the above result to the MMF model introduced in Section \ref{mmf-model-section} by finding the rate distortion regions for the MMF model defined in $(\ref{bob})$ and $(\ref{eve})$ under the two scenarios. In this section, as before, we assume the channels are time-invariant. First of all, we will give the achievable rate region under strong secrecy (therefore also under weak secrecy). 
\begin{thm}\label{channelmmf}
The following rate region for one confidential and one non-confidential message is achievable under strong secrecy for a complex Gaussian channel:
\begin{eqnarray}
R_s&\leq& \log \frac{|HKH^{\dagger}+\sigma_N^2I|}{|\sigma_N^2I|}-\log \frac{|H^eK{H^e}^\dagger+\sigma_{N^e}^2I|}{|\sigma_{N^e}^2I|} \label{rs_gau}\\
R_p&\leq& \log\frac{|HQH^{\dagger}+\sigma_N^2I|}{|HKH^{\dagger}+\sigma_N^2I|} \label{rp_gau}
\end{eqnarray}
for some $K$ and $Q$, where $0\preceq K\preceq Q$, $K\in\mathcal{H}^{M\times M}$, $Q$ satisfies the power constraint in $(\ref{powergen})$, and $H$ and $H^e$ are the channel gain matrices.
\end{thm}
\begin{IEEEproof}
According to Theorem \ref{region} and \ref{thmstrong},
\begin{eqnarray}
R_s&\leq& I(W;Y|V)-I(W;Z|V) \label{rs_gen}\\
R_p&\leq& I(V;Y) \label{rp_gen}
\end{eqnarray}
for some $V\inout W \inout X \inout YZ$ and $\mathbb{E}[XX^{\dagger}] \preceq Q$, is an achievable rate pair.

We restrict the channel input $X$ to be a circularly symmetric complex Gaussian vector. 
Let $V\sim \mathcal{CN}(0,Q-K,0)$, $B\sim \mathcal{CN}(0,K,0)$ such that $B$ and $V$ are independent, and $W=X=V+B$.
Therefore, $X\sim \mathcal{CN}(0,Q,0)$ satisfies the power constraint. Similar to results in \cite{mimo-gaussian}, the rate pair $(R_s, R_p)$ satisfying inequalities $(\ref{rs_gau})$ and $(\ref{rp_gau})$ can be achieved.
\end{IEEEproof}
An immediate corollary follows directly from the above theorem.
\begin{cor} \label{mmfrate}
The following rate pairs are achievable under strong secrecy for MMF with channel gains defined in $(\ref{h})$ and $(\ref{he})$ and equal full power allocation $Q=I$:
\begin{eqnarray}
R_s&\leq& \log \frac{|\text{SNR}K+I|}{|\text{SNR}^e\Psi^eK{\Psi^e}^{\dagger}\Phi+I|} \label{rs_mult}\\
R_p&\leq& \log \frac{|(\text{SNR}+1)I|}{|\text{SNR}K+I|} \label{rp_mult}
\end{eqnarray}
for some $K$ where $0\preceq K\preceq I$, $K\in\mathcal{H}^{M\times M}$, $\text{SNR}=E_0L/\sigma_N^2$ and $\text{SNR}^e=E_0L^e/\sigma_{N^e}^2$.
\end{cor}

With the secrecy capacity region of MMF, we can evaluate its rate distortion region $(R,D)$ under the two extreme cases, without and with causal information at Eve's decoder respectively. For the best case scenario (no causal information), we will give a sufficient condition to force maximum distortion $\Delta$ between Alice and Eve. For the worst case scenario (with causal information), we will give an achievable rate-distortion region and look at the particular case of Hamming distortion.

\begin{thm}\label{sufficient}
For an i.i.d source sequence $S^k$, if 
\begin{equation}
\min_{j\in\{1,...,M\}}\bar{\phi}_j<\frac{\text{SNR}}{\text{SNR}^e}
\label{eq:scsc}
\end{equation} 
where $\bar{\phi}_j$'s are the diagonal entries of $\Phi$,
 then the following rate distortion pair $(R,D)$ is achievable with \textbf{no causal information} at the eavesdropper:
\begin{eqnarray}
R&<&\frac{ M\log(\text{SNR}+1)}{H(S)}\\
D&\leq&\Delta.
\end{eqnarray}
\end{thm}
Theorem \ref{sufficient} follows from Theorem \ref{nocausal} and Corollary \ref{mmfrate}. Note that (\ref{eq:scsc})
is a sufficient condition for the existence of a secure channel with strictly positive rate from Alice to Bob. A discussion of this condition is provided in Appendix \ref{app-suff}.

\begin{thm}\label{mmfrdcausal}
For an i.i.d. source sequence $S^k$ and Hamming distortion, the following distortion rate curve $D(R)$ is in the achievable region \textbf{with causal information} at the eavesdropper:
\begin{eqnarray}
D&=&d(H(S)), \text{ if } R\leq \frac{R_s^*}{H(S)}\\
D&=&\bar{\alpha}(K) \Delta+\left(1-\bar{\alpha}(K)\right) d\left(\frac{R_s(K)}{R}\right), \text{ if } \frac{R_s^*}{H(S)}<R\leq\frac{R_p^*}{H(S)}
\end{eqnarray}
where $d(\cdot)$ is as defined in $(\ref{def_d})$; $\mathcal{K}\triangleq\{K\in\mathcal{H}^{M\times M}, 0\preceq K\preceq I\}$,
$$R_s^*=\max_{K'\in\mathcal{K}}\log\frac{|\text{SNR}K'+I|}{|\text{SNR}^e\sqrt{\Phi}\Psi^eK'{\Psi^e}^{\dagger}\sqrt{\Phi}+I|},$$
$$R_p^*=M\log(\text{SNR}+1),$$
$$R_s(K)=\log\frac{|\text{SNR}K+I|}{|\text{SNR}^e\sqrt{\Phi}\Psi^eK{\Psi^e}^\dagger\sqrt{\Phi}+I|},$$
$$\bar{\alpha}(K)=\frac{\bar{\beta}(K)-\bar{\gamma}(K)}{\bar{\beta}(K)},$$
$$\bar{\beta}(K)=\log\frac{|(\text{SNR}+1)I|}{|\text{SNR}K+I|},$$
$$\bar{\gamma}(K)=\log\frac{|\text{SNR}^e\Phi+I|}{|\text{SNR}^e\sqrt{\Phi}\Psi^eK{\Psi^e}^\dagger\sqrt{\Phi}+I|}.$$
\end{thm}
The result given in Theorem \ref{mmfrdcausal} can be derived directly from Theorem \ref{rdcausal} and Corollary \ref{mmfrate}.

\subsection{Secrecy Outage under Random MMF Channel} \label{mmf-random}
All the results we have seen thus far were derived for a time-invariant channel, which means that both the transmitter and the receivers are informed about the channel state. 
However, in MMF, the channels ${H}$ and ${H}^e$ vary with time and the CSI is not available at the transmitter due to the long round-trip delay over the large distances common in optical transmission, even though it has a long coherence time, i.e. $H$ and $H^e$ are essentially constant over $n$ channel uses. In Corollary \ref{cor-causal}, Theorem \ref{sufficient} and Theorem \ref{mmfrdcausal}, we have chosen $Q=I$ as the channel input power for simplicity. This power allocation strategy also tends to minimize the outage probability for perfect secrecy \cite{asilomar}.
 
 The randomness of $H=\sqrt{E_0L}\Psi$ and $H^e=\sqrt{E_0L^e}\sqrt{\Phi}\Psi^e$ comes from the unitary component $\Psi$, $\Psi^e$ and the diagonal component $\Phi$. The random matrices $\Psi$ and $\Psi^e$ are uniformly distributed in $\mathcal{\Psi}$, where $\mathcal{\Psi}$ is the set of all $M\times M$ unitary matrices \cite{winzer-express}. 
The diagonal matrix $\Phi=diag\{\bar\phi_1, ..., \bar\phi_M\}$, where $\bar\phi_i=M\frac{\phi_i}{\sum_{j=1}^M \phi_j}$. Here $\phi_1=\phi_{min}$ and $\phi_{max}$, and $\phi_i \sim Unif[\phi_{min},\phi_{max}]$ for $i=3,...,M$.

In this situation of no CSI at the transmitter with long coherence time, performance is typically measured by outage probability. The capacity $C=\max_Q \log|I+HQH^{\dagger}|$ and secrecy capacity $C_s=\max_Q \left[\log|I+HQH^{\dagger}|-\log|I+H^eQ{H^e}^{\dagger}|\right]$ for a deterministic MIMO Gaussian broadcast channel were given in \cite{emre} and \cite{ogg}, respectively. For the case of no causal information at the eavesdropper, the CSI does not really affect the performance much. The channel capacity between Alice and Bob $M\log(1+\text{SNR})$ does not depend on the channel realization due to the unitary component of the channel. Hence, the encoder can choose the source-channel coding rate to be just below $\frac{M\log(\text{SNR}+1)}{H(S)}$, and maximum distortion can be achieved if the channel satisfies the condition in Theorem \ref{sufficient}. For the case where Eve has causal source information, we consider only the input power $Q=I$ and Hamming distortion. Without knowledge of CSI, the source-channel encoder picks a source-channel coding rate $\bar{R}$, a pair of source coding rates $(\bar{R}_s', \bar{R}_p')$ on the line segment $R_s'+R_p'=H(S)$, $R_s', R_p' \geq 0$ and a real value $\alpha\in[0,1]$. Let 
$$\bar{R}_s \triangleq \bar{R}_s'\bar{R} \text{ and } \bar{R}_p\triangleq \bar{R}_p'\bar{R}.$$ 
We define the outage probability of such a choice of parameters to be
$$P_{out}(I, \bar{R}_s', \bar{R},\alpha)=1-\sum_{K\in\mathcal{K}}\mathbb{P}_{\Phi \Psi^e\Psi}\left[(\bar{R}_s, \bar{R}_p)\in \mathcal{R}_{\Phi \Psi^e\Psi}(I) \text{ and } \bar{\alpha}(K)\geq \alpha \right],$$
where $\mathcal{R}_{\Phi \Psi^e\Psi}(I)$ denotes the region given in Corollary \ref{mmfrate},  $\mathcal{K}=\{K: 0\preceq K \preceq I, \bar{R}_s=\log \frac{|\text{SNR}K+I|}{|\text{SNR}^e\Psi^eK{\Psi^e}^{\dagger}\Phi+I|} \text{ and } \bar{R}_p= \log \frac{|(\text{SNR}+1)I|}{|\text{SNR}K+I|}\}$,and $\bar{\alpha}(K)$ is defined as in Theorem \ref{mmfrdcausal}. Note that $\mathcal{R}_{\Phi \Psi^e\Psi}(Q)$ is a random variable because the channels $P_{YZ|X}$ are now random. Under this set of parameters $(\bar{R}_s', \bar{R},\alpha)$, we can achieve distortion $\alpha\Delta+(1-\alpha) d(\bar{R}_s')$ with probability $1-P_{out}(I, \bar{R}_s', \bar{R},\alpha)$, where $d(\cdot)$ was defined in $(\ref{def_d})$. Proving the existence of a good codebook in this case is an important information theoretic problem, most recently addressed in \cite{existence}.

\section{Numerical Results} \label{numerical}
In this section, we present numerical results illustrating achievable rate distortion regions of an MMF under the two information models with a time-invariant channel. Let us consider measuring the eavesdropper's distortion using  Hamming distortion and a Bern($p$) i.i.d. source sequence. Fig. 3 shows numerical results corresponding to Theorem \ref{sufficient} and Theorem \ref{mmfrdcausal} under equal power allocation. The channels are simulated as a $4-$mode MMF with $\text{SNR}=20dB$, $\text{SNR}^e=10dB$, and $\text{MDL}=20dB$. 

In each plot, the vertical line on the right is the maximum reliable transmission rate between Alice and Bob and the vertical line on the left is the maximum perfect secrecy transmission rate that can be obtained with separate source-channel coding. The horizontal line is the maximum distortion which is also the rate distortion curve from Theorem \ref{sufficient} with no causal information at Eve. The curve obtained from Theorem \ref{mmfrdcausal} shows the tradeoff between the transmission rate between Alice and Bob and the distortion forced on Eve with causal information. We see in Fig. \ref{rdcurve}(a), $p=0.3$,  that with our source-channel coding analysis, we gain a free region for maximum distortion, as if under perfect secrecy, (from the left vertical line to the kink) because we effectively use the redundancy of the source. In Fig. \ref{rdcurve}(b) with $p=0.5$, since there is no redundancy in the source, the distortion curve drops immediately after the maximum perfect secrecy rate. Note that the transmission rates are not considered beyond the right vertical lines because they are above the maximum reliable transmission rates and Bob cannot losslesly reconstruct the source sequences.
\begin{figure}[htbp]
  \centering
  \includegraphics[width=18 cm]{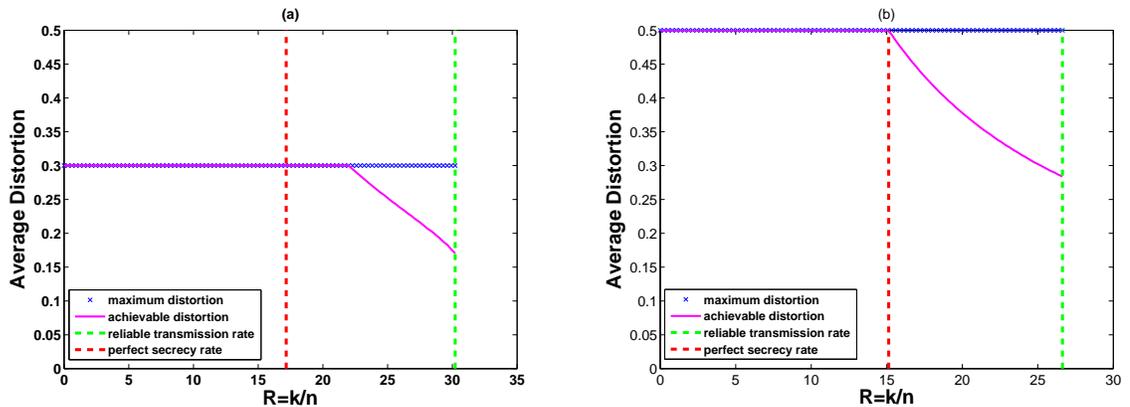}
\caption{Achievable distortion-rate curves. On the left is the Bern(0.3) i.i.d. source case and on the right is the Bern(0.5) i.i.d. source case. On the horizontal axes are the symbol/channel use source-channel coding rate and on the vertical axes are the average Hamming distortions.} 
\label{rdcurve}
\end{figure}


\section{Conclusion} \label{conclusion}
In this work, we have examined the rate-distortion-based secrecy performance of an insecure MMF communication system. The sender is assumed to have an i.i.d. source sequence which the intended receiver and the eavesdropper both try to reconstruct. Two source-channel coding models with different information availability at the eavesdropper have been considered. We have shown that, when no causal source information is disclosed to the eavesdropper, under a general broadcast channel and any distortion measure, it is possible to send the source at the maximum rate that guarantees lossless reconstruction at the intended receiver while keeping the distortion at the eavesdropper as high as if it only has the source prior distribution. When the past source realization is causally disclosed to the eavesdropper, we have applied the theoretical results in \cite{allerton} to the particular case of an MMF channel. Numerical results for an i.i.d. Bernoulli source and Hamming distortion have been provided. 

Only the theoretical formulation is given to calculate the secrecy outage probability for random MMF under equal full power allocation $Q=I$. The optimality of this power strategy and the statistics for different sets of parameters given in Section \ref{mmf-random} remain an open problem. Moreover, in our model, it is required that the intended receiver reconstruct the source losslessly. In a more general setting, one can allow lossy reconstruction of the source at the intended receiver, which is an interesting problem for further research. 

\newpage
\appendices
\section{Proof of Lemma \ref{lem1}} \label{applem1}
Given $(m_s,m_p)$, suppose there exists $\theta_{m_p}$ such that 
\begin{align}
||P_{Z^n|M_p=m_p, M_s=m_s}-\theta_{m_p}||_{TV}&\leq \epsilon_n \label{suff}
\end{align}
where $\epsilon_n=2^{-n\beta}$ for some $\beta>0$.
Then we have the following:
\begin{eqnarray}
&&||P_{Z^n|M_p=m_p}-\theta_{m_p}||_{TV}\notag\\
&=&\sum_{z^n}\left|P_{Z^n|M_p=m_p}(z^n)-\theta_{m_p}(z^n)\right|\\
&=&\sum_{z^n}\left| \sum_{m_s}P_{M_s|M_p=m_p}(m_s)P_{Z^n|M_p=m_p,M_s=m_s}(z^n)-\sum_{m_s}P_{M_s|M_p=m_p}(m_s)\theta_{m_p}(z^n)\right| \notag\\
&=&\sum_{z^n}\left|\sum_{m_s}\frac{1}{|\mathcal{M}_s|}P_{Z^n|M_p=m_p,M_s=m_s}(z^n)-\sum_{m_s}\frac{1}{|\mathcal{M}_s|}\theta_{m_p}(z^n)\right|\notag\\
&\leq& \sum_{z^n}\sum_{m_s}\frac{1}{|\mathcal{M}_s|}\left|P_{Z^n|M_p=m_p,M_s=m_s}(z^n)-\theta_{m_p}(z^n)\right| \label{tri}\\
&=&\sum_{m_s}\frac{1}{|\mathcal{M}_s|}\sum_{z^n}\left|P_{Z^n|M_p=m_p,M_s=m_s}(z^n)-\theta_{m_p}(z^n)\right|\notag\\
&\leq& \sum_{m_s} \frac{1}{|\mathcal{M}_s|}\epsilon_n \label{tv}\\
&=&\epsilon_n
\end{eqnarray}
where $(\ref{tri})$ follows from triangle inequality and $(\ref{tv})$ follows from $(\ref{suff})$.

\begin{eqnarray*}
&&||P_{Z^n|M_p=m_p}P_{M_s|M_p=m_p}-P_{Z^nM_s|M_p=m_p}||_{TV}\\
&=&\sum_{z^n}\sum_{m_s}\left|P_{Z^n|M_p=m_p}(z^n)P_{M_s|M_p=m_p}(m_s)-P_{Z^n|M_p=m_p,M_s=m_s}(z^n)P_{M_s|M_p=m_p}(m_s)\right|\\
&=&\frac{1}{|\mathcal{M}_s|}\sum_{z^n}\sum_{m_s}\left|P_{z^n|M_p=m_p}(z^n)-P_{Z^n|M_p=m_p,M_s=m_s}(z^n)\right|\\
&=&\frac{1}{|\mathcal{M}_s|}\sum_{z^n}\sum_{m_s}\left|P_{Z^n|M_p=m_p}(z^n)-\theta_{m_p}(z^n)+\theta_{m_p}(z^n)-P_{Z^n|M_p=m_p,M_s=m_s}(z^n)\right|\\
&\leq& \frac{1}{|\mathcal{M}_s|} \sum_{z^n}\sum_{m_s}(\left|P_{Z^n|M_p=m_p}(z^n)-\theta_{m_p}(z^n)\right|+\left|P_{Z^n|M_p=m_p,M_s=m_s}(z^n)-\theta_{m_p}(z^n)\right|)\\
&=&\frac{1}{|\mathcal{M}_s|}\sum_{m_s}(\sum_{z^n}\left|P_{Z^n|M_p=m_p}(z^n)-\theta_{m_p}(z^n)\right|+\sum_{z^n}\left|P_{Z^n|M_p=m_p,M_s=m_s}(z^n)-\theta_{m_p}(z^n)\right|)\\
&\leq& \frac{1}{|\mathcal{M}_s|}\sum_{m_s}(\epsilon_n+\epsilon_n)\\
&=&2\epsilon_n
\end{eqnarray*}
By applying Lemma \ref{lem2}, we have
\begin{eqnarray}
I(M_s;Z^n|M_p)&=&\sum_{m_p}P_{M_p}(m_p)I(M_s;Z^n|M_p=m_p)\notag\\
&\leq& \sum_{m_p}P_{M_p}(m_p) (-2\epsilon_n\log \frac{2\epsilon_n}{|\mathcal{M}_s|})\notag\\
&\leq& 2\cdot2^{-n\beta}(nR_s) \label{zero}
\end{eqnarray}
where $(\ref{zero})$ goes to $0$ as $n\rightarrow \infty$.

\section{Proof of $(\ref{resultfin})$} \label{detail}
For each $i$, we have
\begin{eqnarray}
I(S_i;Z^n|M_p)&\leq& I(M_sS_i;Z^n|M_p)\notag\\
&=& I(M_s;Z^n|M_p)+ I(S_i;Z^n|M_sM_p)\notag\\
&\leq& \epsilon \label{strongs}
\end{eqnarray}
for large enough $n$. $(\ref{strongs})$ follows from strong secrecy of the channel and Fano's inequality. Note that weak secrecy is not sufficient to give us the  desired result in our proof. We now define
$$P_i\triangleq P_{S_iZ^nM_p}$$
$$\bar{P}_i\triangleq P_{M_p}P_{S_i|M_p}P_{Z^n|M_p}$$
i.e. $\bar{P}_i$ is the Markov chain $S_i\inout M_p\inout Z^n$.  By Pinsker's inequality, 
\begin{eqnarray}
||P_i-\bar{P}_i||_{TV}&\leq&\frac{1}{\sqrt{2}} D(P_i||\bar{P}_i)^{\frac12}\notag\\
&=&\frac{1}{\sqrt{2}} I(S_i;Z^n|M_p)^{\frac12}\notag\\
&\leq& \sqrt{\frac{\epsilon}{2}} \label{totalv}
\end{eqnarray}
\begin{eqnarray}
\min_{t(i,z^n)}\mathbb{E}[d(S_i,t(i,Z^n))]&\geq& \min_{t(i,z^n,m_p)}\mathbb{E}[d(S_i,t(i,Z^n,M_p))]\notag\\
&\geq&\min_{t(i,z^n,m_p)}\mathbb{E}_{\bar{P}_i}[d(S_i,t(i,Z^n,M_p))]-\delta^\prime(\epsilon)\label{ptopbar}\\
&=&\min_{t(i,m_p)}\mathbb{E}_{\bar{P}_i}[d(S_i,t(i,M_p))]-\delta^\prime(\epsilon)\label{mkv}\\
&\geq&\min_{t(i,m_p)}\mathbb{E}[d(S_i,t(i,M_p))]-2\delta^\prime(\epsilon)\label{pbartop}
\end{eqnarray}
where $(\ref{ptopbar})$ and $(\ref{pbartop})$ use the fact that $P_i$ and $\bar{P}_i$ are close in total variation from $(\ref{totalv})$; $(\ref{mkv})$ uses the Markov relation $S_i\inout M_p\inout Z^n$ of distribution $\bar{P}_i$. The technical details can be found in Lemma 2 and 3 from \cite{allerton}. Averaging over $k$, we obtain $(\ref{resultfin})$.

\section{Justification of the condition $\max_{(R_s,R_p)\in\mathcal{R}}R_s>0$} \label{just}
From Theorem \ref{region} or \ref{thmstrong}, we have 
$$\max_{(R_s,R_p)\in\mathcal{R}}R_s>0$$ 
is equivalent to  
\begin{eqnarray}
I(W;Y|V)-I(W;Z|V)&>&0 \label{unreduced}
\end{eqnarray}
for some $V\inout W\inout X \inout YZ$. We claim this can be simplified to 
\begin{eqnarray}
I(W;Y)-I(W;Z)&>&0 \label{reduced}
\end{eqnarray}
for some $W\inout X\inout YZ$.

To see $(\ref{reduced})\Rightarrow (\ref{unreduced})$, we can simply let $V=\o$.
To see $(\ref{unreduced})\Rightarrow (\ref{reduced})$, observe that if there exists $V\inout W\inout X\inout YZ$ such that $(\ref{unreduced})$ holds, then there has to exist at least one value $v$ such that 
$I(W;Y|V=v)-I(W;Z|V=v)>0$. We can redefine the distribution as $P_{\bar{W}\bar{X}\bar{Y}\bar{Z}}\triangleq P_{WXYZ|V=v}$. It can be verified that the Markovity $\bar{W}\inout \bar{X} \inout \bar{Y}\bar{Z}$ holds and $P_{\bar{Y}\bar{Z}|\bar{X}}=P_{YZ|X}$.

\section{Sufficient condition on Theorem \ref{sufficient}} \label{app-suff}
From Theorem \ref{nocausal} and Corollary \ref{mmfrate}, we know that a sufficient condition for the eavesdropper's channel not being less noisy than the intended receiver's channel is 
\begin{equation}
\max_{K\in \mathcal{H}^{M\times M}, 0\preceq K\preceq I} \frac{|\text{SNR}K+I|}{|\text{SNR}^e\Psi^eK{\Psi^e}^{\dagger}\Phi+I|}>1. \label{gencond}
\end{equation}
However, $(\ref{gencond})$ is computationally heavy to verify. If we restrict $K$ to be of the form $K={\Psi^e}^{\dagger}\Lambda\Psi^e$ where $\Lambda$ is diagonal with diagonal entries $\lambda_i\in[0,1]$, then $(\ref{gencond})$ has a much simpler form:
\begin{equation}
\frac{\prod_{i=1}^M(1+\text{SNR}\lambda_i)}{\prod_{i=1}^M(1+\text{SNR}^e\lambda_i\bar{\phi}_i)}>1 \label{redcond}.
\end{equation}
Therefore, if there exists a $j\in\{1,...,M\}$ such that $\bar{\phi}_j<\frac{\text{SNR}}{\text{SNR}^e}$, we can choose $\lambda_j=1$ and $\lambda_i=0$ for $i\neq j$ to satisfy $(\ref{redcond})$.

\section*{Acknowledgment}

The authors would like to thank Dr. Peter Winzer from Bell Labs, Alcatel Lucent, Dr. Matthieu Bloch and Dr. Rafael Schaefer for fruitful discussions and great support on this project.

\ifCLASSOPTIONcaptionsoff
  \newpage
\fi

\bibliographystyle{ieeetr}

\bibliography{tcom_submission_final}

\end{document}